\pdfoutput=1
\documentclass[aps,pra,reprint,superscriptaddress,nofootinbib,floatfix]{revtex4-2}

\usepackage{amsmath,amssymb,amsthm,bm}
\usepackage{booktabs}
\usepackage{graphicx}
\usepackage{microtype}
\usepackage{xcolor}
\usepackage[colorlinks=true,linkcolor=blue,citecolor=blue,urlcolor=blue]{hyperref}

\newcommand{\Cl}{\mathrm{Cl}}
\newcommand{\C}{\mathbb{C}}
\newcommand{\Tr}{\operatorname{Tr}}
\newcommand{\ii}{\mathrm{i}}
\newcommand{\ket}[1]{\lvert #1\rangle}
\newcommand{\bra}[1]{\langle #1\rvert}
\newcommand{\avg}[1]{\langle #1\rangle}
\newcommand{\calP}{\mathcal{P}}
\newcommand{\supp}{\operatorname{supp}}

\newtheorem{lemma}{Lemma}
\newtheorem{proposition}{Proposition}

\begin{document}

\title{Operator-centric Clifford algebra for variational eigensolvers and finite-shot adaptive selection}

\author{Ginanjar Utama}
\email{ginanjar.utama@gmail.com}
\affiliation{Department of Engineering Physics, Institut Teknologi Bandung, Bandung, Indonesia}

\author{Hermawan Kresno Dipojono}
\email{dipojono@itb.ac.id}
\affiliation{Department of Engineering Physics, Institut Teknologi Bandung, Bandung, Indonesia}

\date{July 19, 2026}

\begin{abstract}
We develop a sparse operator-centric realization of $n$-qubit variational quantum algorithms in the complex Clifford algebra $\Cl(2n,\C)\cong M(2^n,\C)$.  Density operators, gates, observables, channels, fermionic modes, and adaptive-selection observables are represented in one Pauli-word algebra, with the Jordan--Wigner map providing the exact bridge to anticommuting Clifford generators.  We distinguish general Pauli-word rotations from Spin-group rotors and formulate the familiar odd-$Y$ restriction for real-state adaptive ansatzes as an exact transpose-parity statement: for real Hamiltonians and real states, every candidate Pauli word containing an even number of $Y$ factors has zero ADAPT gradient, while odd-$Y$ rotations preserve the real sector.  For the critical open transverse-field Ising chain, a depth-three Hamiltonian variational ansatz gives relative energy errors $4.84\times10^{-5}$, $2.19\times10^{-3}$, and $3.67\times10^{-3}$ for $n=4,5,6$.  A compact local ADAPT pool is exact at $n=4$ but leaves residual errors at larger sizes; a systematic contiguous three-local odd-$Y$ pool reaches relative errors below $1.3\times10^{-12}$ for $n\leq6$.  In 100-seed finite-shot tests at $n=4$, fixed-shot selection succeeds in $0/100$ runs, whereas uniform escalation and confidence-bound racing each succeed in $84/100$ runs; racing lowers median shots by $34\%$.  We claim no asymptotic speedup over matrix methods.  The contribution is a corrected algebraic formulation, a density-operator derivation and implementation of the real-sector pool filter, and a reproducible study of measurement-limited adaptive selection.
\end{abstract}

\maketitle

\section{Introduction}
\label{sec:intro}

The variational quantum eigensolver (VQE) combines a parameterized quantum state with a classical optimization loop to approximate low-energy eigenstates \cite{peruzzo2014vqe}.  Its practical performance depends strongly on the ansatz.  Fixed-depth circuits are predictable but may lack problem-specific expressivity, whereas adaptive methods such as ADAPT-VQE grow a circuit by selecting the operator with the largest instantaneous energy gradient \cite{grimsley2019adapt,tang2021qubitadapt}.  Analytic circuit gradients remove finite-difference bias \cite{mitarai2018qcl,schuld2019gradients,wierichs2022general}, but adaptive selection introduces a different bottleneck: the pool gradients must themselves be ranked from finite measurements.

Several approaches reduce this overhead.  Commuting-gradient measurement, adaptive sampling, subpool exploration, measurement reuse, variance-aware allocation, and best-arm successive elimination all target different parts of the cost \cite{anastasiou2023gradients,majland2023fast,long2024layering,ikhtiarudin2025shot,huang2025bestarm}.  Shot noise remains consequential because a wrong ranking changes the ansatz before parameter optimization begins \cite{scriva2024challenges}.  The present work isolates that ranking problem while keeping parameter reoptimization exact.

Our second motivation is algebraic.  Geometric and Clifford algebras have long been used to represent qubits, density operators, entanglement, and quantum gates \cite{somaroo1998operations,havel2000ga,hrdina2022clifford,cafaro2024revisited}.  A recent operator-centric exposition emphasizes the full complex algebra rather than a selected minimal left ideal \cite{silva2025operator}.  We adopt that operator picture, but make the multi-qubit map explicit: bare Pauli operators on distinct qubits commute, so they cannot be identified directly with mutually anticommuting Clifford generators.  The Jordan--Wigner strings are mandatory, not optional.

This paper does not argue that rewriting matrices as multivectors produces an asymptotic speedup.  Since $\Cl(2n,\C)$ and $M(2^n,\C)$ are isomorphic, a dense element has $4^n$ complex coefficients in either language.  Instead, we use the sparse Pauli-word realization to place state evolution, backward observable propagation, commutator gradients, and measurement statistics in one typed algebra.  At the implementation layer, this storage model is closely related to specialized Pauli-arithmetic backends.  PauliEngine supplies a binary-symplectic C++ core with Python bindings for multiplication, commutators, phase tracking, symbolic coefficients, substitutions, and structural transformations \cite{muller2026pauliengine}.  PauLIB emphasizes bit packing, SIMD-friendly bulk operations, and multithreaded Pauli-sum processing \cite{krotz2026paulib}, while Li \emph{et al.} combine binary symplectic encoding with grouped sparse representations in a Julia/C++ framework for quantum-chemical and many-body workloads \cite{li2026pauliframework}.  These projects are neighboring infrastructure for the sparse arithmetic layer; they do not by themselves provide the Clifford-blade dictionary or the adaptive-selection analysis below.  Conversely, our compact Python kernel is intended as a transparent reference implementation, not as a performance competitor to those systems.

The contribution is fourfold:
\begin{enumerate}
    \item We give a phase-consistent Pauli-word/Clifford-blade dictionary and distinguish Pauli-word rotations, Spin rotors, and the quantum-information Clifford group.
    \item We cast the established odd-$Y$ restriction for real-state qubit-ADAPT as an exact density-operator transpose-parity rule, including an inductive real-sector preservation statement.
    \item We derive exact gradients for shared-parameter Hamiltonian variational ansatzes as sums of physical-gate contributions, clarifying why a naive two-point shift of an aggregate parameter can fail even when the gates within a layer commute.
    \item We compare fixed-shot ranking, uniform cumulative escalation, and confidence-bound racing at a common shot ceiling, using 100 independent seeds and Wilson confidence intervals.
\end{enumerate}

We benchmark the framework on the open transverse-field Ising model (TFIM),
\begin{equation}
H(J,h)=-J\sum_{j=0}^{n-2}Z_jZ_{j+1}-h\sum_{j=0}^{n-1}X_j,
\label{eq:tfim}
\end{equation}
whose thermodynamic critical point is $h/J=1$ \cite{pfeuty1970ising}.

\section{Operator-centric Clifford formulation}
\label{sec:clifford}

\subsection{One operator algebra}

Let $\gamma_0,\ldots,\gamma_{2n-1}$ generate the complex Clifford algebra $\Cl(2n,\C)$,
\begin{equation}
\gamma_a\gamma_b+\gamma_b\gamma_a=2\delta_{ab}\mathbf{1}.
\label{eq:clifford}
\end{equation}
The standard irreducible representation gives
\begin{equation}
\Cl(2n,\C)\cong M(2^n,\C).
\label{eq:isomorphism}
\end{equation}
The operator-centric picture therefore represents every $n$-qubit operator as one algebra element.  A state is a density multivector $\rho=\rho^\dagger\geq0$ with $\Tr\rho=1$; a closed-system gate is a unitary multivector $U$; an observable is Hermitian; and a completely positive trace-preserving map has Kraus elements in the same algebra:
\begin{align}
\rho'&=U\rho U^\dagger,\\
\avg{O}_\rho&=\Tr(O\rho),\\
\mathcal{E}(\rho)&=\sum_k K_k\rho K_k^\dagger,
&\sum_kK_k^\dagger K_k&=\mathbf{1}.
\end{align}
This representation is not more expressive in principle than a spinor space together with its endomorphism algebra.  Its advantage is uniformity: mixed states, channels, observables, and selection gradients do not require a change of mathematical type.  A minimal left ideal remains available when a ket-like object is preferable \cite{hrdina2022clifford,lounesto2001clifford}.

\subsection{Jordan--Wigner generators and phase-normalized blades}

A direct local assignment $\gamma_{2j}\mapsto X_j$, $\gamma_{2j+1}\mapsto Y_j$ fails for $n>1$, because Pauli operators on different qubits commute.  We use the Jordan--Wigner representation \cite{jordan1928}
\begin{align}
\gamma_{2j}&=Z_0Z_1\cdots Z_{j-1}X_j,\label{eq:jwx}\\
\gamma_{2j+1}&=Z_0Z_1\cdots Z_{j-1}Y_j.\label{eq:jwy}
\end{align}
The fermionic operators
\begin{equation}
c_j=\frac{\gamma_{2j}+\ii\gamma_{2j+1}}{2},\qquad
c_j^\dagger=\frac{\gamma_{2j}-\ii\gamma_{2j+1}}{2}
\end{equation}
then obey the canonical anticommutation relations and $c_j^\dagger c_j=(1-Z_j)/2$.

A raw ordered blade $\gamma_{a_1}\cdots\gamma_{a_k}$ is not necessarily Hermitian.  For $A=\{a_1<\cdots<a_k\}$ we therefore define the phase-normalized blade
\begin{equation}
\widehat{\gamma}_A=
\ii^{k(k-1)/2}\gamma_{a_1}\cdots\gamma_{a_k}.
\label{eq:hermitianblade}
\end{equation}
It satisfies $\widehat{\gamma}_A^\dagger=\widehat{\gamma}_A$ and $\widehat{\gamma}_A^2=\mathbf{1}$.  Equations~\eqref{eq:jwx}--\eqref{eq:jwy} map these $4^n$ phase-normalized blades bijectively, up to an overall sign, to the Hermitian Pauli words
\begin{equation}
\calP_n=\{I,X,Y,Z\}^{\otimes n}.
\end{equation}
We store an operator sparsely as
\begin{equation}
A=\sum_{W\in\calP_n}a_WW.
\label{eq:pauliexpansion}
\end{equation}
Because nonidentity Pauli words are traceless,
\begin{equation}
\Tr A=2^n a_I=2^n\avg{A}_0,
\label{eq:trace}
\end{equation}
where $\avg{A}_0$ denotes the identity-word coefficient.  The Hilbert--Schmidt product is $\Tr(A^\dagger B)=2^n\sum_Wa_W^*b_W$.  The active support $|\supp(A)|$ is the number of nonzero coefficients in Eq.~\eqref{eq:pauliexpansion}; sparse multiplication and coefficient intersection are used until support growth makes dense linear algebra preferable.

\subsection{Pauli-word rotations, Spin rotors, and Clifford gates}
\label{sec:rotations}

For a Hermitian Pauli word $P$ with $P^2=\mathbf{1}$,
\begin{equation}
R_P(\theta)=e^{-\ii\theta P/2}
=\cos\frac{\theta}{2}\,\mathbf{1}
-\ii\sin\frac{\theta}{2}\,P,
\label{eq:paulirotation}
\end{equation}
with derivative
\begin{equation}
\partial_\theta R_P(\theta)=-\frac{\ii}{2}P R_P(\theta).
\label{eq:rotationderivative}
\end{equation}
We call Eq.~\eqref{eq:paulirotation} a \emph{Pauli-word rotation}.  With respect to the chosen Euclidean real form, it is a Spin-group rotor only when the Lie-algebra generator $-\ii P$ is a bivector.  Even Clifford grade alone is insufficient: for example, $Z_j$ corresponds to a bivector, whereas $Z_jZ_{j+1}$ corresponds to a grade-four blade under the chosen Jordan--Wigner convention.

The ``Clifford group'' in quantum information is a different object: it is the normalizer of the Pauli group \cite{aaronson2004stabilizer,tolar2018clifford}.  A Clifford gate maps each Pauli word to one Pauli word up to phase under conjugation, while a generic Pauli-word rotation can spread support.  The terminology is kept separate throughout.

For sparse propagation, conjugation by Eq.~\eqref{eq:paulirotation} is evaluated word by word.  If $P$ commutes with a word $W$, it is unchanged; if $P$ anticommutes with $W$,
\begin{equation}
R_P(\theta)W R_P^\dagger(\theta)
=\cos\theta\,W-\ii\sin\theta\,PW.
\label{eq:fastconjugation}
\end{equation}
Thus one Pauli rotation at most doubles the support of an operator.

\subsection{Transpose parity and the real sector}
\label{sec:reality}

For a Pauli word $W$, let $\nu_Y(W)$ be the number of $Y$ factors.  Computational-basis transposition acts as
\begin{equation}
W^T=(-1)^{\nu_Y(W)}W.
\label{eq:transposeparity}
\end{equation}
This is a decomposition by \emph{transpose parity}, not an associative-algebra grading, because transposition reverses product order: $(AB)^T=B^TA^T$.

Odd-$Y$ Pauli pools are established in real-valued qubit-ADAPT calculations \cite{tang2021qubitadapt}, and the vanishing of even-$Y$ gradients for real Hamiltonians and wave functions has been stated explicitly in prior work \cite{mukherjee2023impurity}.  We include the argument to express that restriction at the density-operator level, connect it directly to transposition in the sparse basis, and make preservation of the real sector explicit; it should not be read as the first use of odd-$Y$ pools.

\begin{lemma}[Real-sector preservation]
\label{lem:realpreservation}
Let $\rho$ be real symmetric and let $P$ be a Pauli word with odd $\nu_Y(P)$.  Then $R_P(\theta)\rho R_P^\dagger(\theta)$ is real symmetric for every real $\theta$.
\end{lemma}
\begin{proof}
An odd-$Y$ Hermitian word is purely imaginary and antisymmetric.  Therefore $-\ii P$ is real antisymmetric, $R_P(\theta)$ is real orthogonal, and $R_P\rho R_P^T$ is real symmetric.
\end{proof}

\begin{proposition}[Transpose-parity selection rule]
\label{prop:selectionrule}
Let $H$ and $\rho$ be real symmetric.  If a Hermitian Pauli word $P$ has even $\nu_Y(P)$, then
\begin{equation}
\Tr\!\left(\rho[H,P]\right)=0.
\label{eq:zerogradient}
\end{equation}
\end{proposition}
\begin{proof}
An even-$Y$ word is real symmetric.  Hence $[H,P]^T=-[H,P]$ is real antisymmetric.  The trace of the product of a symmetric and an antisymmetric matrix vanishes.
\end{proof}

For a real Hamiltonian and a real reference state, Lemma~\ref{lem:realpreservation} makes Proposition~\ref{prop:selectionrule} inductive: an ansatz built from odd-$Y$ generators stays real, and every even-$Y$ candidate has zero gradient at every ADAPT step.  Among all $4^n$ Pauli words, the odd-$Y$ sector contains
\begin{equation}
N_{\mathrm{odd}\,Y}=\frac{4^n-2^n}{2}
\label{eq:oddcount}
\end{equation}
words.  The theorem prunes the complementary sector exactly, but does not by itself determine which local odd-$Y$ pool is sufficiently expressive.

\section{Variational algorithms}
\label{sec:algorithms}

\subsection{Hamiltonian variational ansatz}

We initialize the TFIM in
\begin{equation}
\rho_+=\ket{+\cdots+}\bra{+\cdots+},
\end{equation}
which is the exact $J=0$ ground state.  A depth-$p$ Hamiltonian variational ansatz (HVA) alternates interaction and field rotations,
\begin{equation}
U(\bm\theta)=\prod_{\ell=1}^{p}
\left[
\prod_{j=0}^{n-1}R_{X_j}(\beta_\ell)
\prod_{j=0}^{n-2}R_{Z_jZ_{j+1}}(\gamma_\ell)
\right],
\label{eq:hva}
\end{equation}
with the rightmost factor applied first.  Each layer has two shared parameters and the objective is
\begin{equation}
E(\bm\theta)=\Tr\!\left[H U(\bm\theta)\rho_+U^\dagger(\bm\theta)\right].
\end{equation}

The two-point parameter-shift rule is exact for one physical gate generated by a two-eigenvalue Pauli word.  It is not generally exact when one shifts a shared parameter controlling an entire layer.  In the present HVA the gates within each field block commute with one another, as do those within each interaction block; nevertheless, the aggregate generator has more than two eigenvalues.  The exact derivative of shared parameter $\theta_s$ is the sum of physical-gate contributions.  Let $G_k=R_{P_k}(\theta_{s(k)})$, let $\rho_k$ be the state immediately after gate $k$, and let $O_k$ be $H$ propagated backward through all later gates.  Then
\begin{equation}
\frac{\partial E}{\partial\theta_s}
=\sum_{k:s(k)=s}\operatorname{Re}\left[
-\frac{\ii}{2}\Tr\!\left(O_k[P_k,\rho_k]\right)
\right].
\label{eq:adjointgradient}
\end{equation}
A forward state sweep and backward observable sweep compute all shared gradients.  In the code, Eq.~\eqref{eq:adjointgradient} is checked against both finite differences and the sum of individual-gate parameter shifts.

\subsection{ADAPT selection as an observable}

Suppose a candidate $R_P(\theta)$ is appended to the current state $\rho$.  At $\theta=0$,
\begin{align}
g_P
&=\left.\frac{\partial}{\partial\theta}
\Tr\!\left[H R_P(\theta)\rho R_P^\dagger(\theta)\right]
\right|_{\theta=0}\\
&=-\frac{\ii}{2}\Tr\!\left(\rho[H,P]\right)
=\Tr(G_P\rho),
\label{eq:adaptgradient}\\
G_P&=-\frac{\ii}{2}[H,P]=G_P^\dagger.
\label{eq:selectionobservable}
\end{align}
Thus ranking candidates does not require candidate-specific shifted circuits; on hardware it requires repeated preparation of the current ansatz state and measurement of the Pauli words in $G_P$.

We use two open-chain, no-repeat pools.  The compact pool is
\begin{equation}
\mathcal{A}_{\mathrm{ZYZ}}=
\left\{\left(\prod_{j\in S}Z_j\right)Y_i:
S\subseteq\{i-1,i+1\}\right\},
\label{eq:compactpool}
\end{equation}
with boundary-valid neighbors only.  It has $4n-4$ elements for $n\geq2$.  The systematic pool $\mathcal{A}_{\mathrm{local3}}$ contains every contiguous one-, two-, and three-site word over $X,Y,Z$ with odd $Y$ parity.  Its open-chain size is
\begin{equation}
|\mathcal{A}_{\mathrm{local3}}|
=\sum_{\ell=1}^{3}(n-\ell+1)\frac{3^\ell-1}{2}.
\label{eq:local3size}
\end{equation}
At each step, the largest $|g_P|$ is appended at zero angle and all parameters are reoptimized.  Because the old optimum remains feasible, exact reoptimization guarantees nonincreasing optimized energy and preserves the variational upper-bound property $E(\bm\theta)\geq E_0$.  This monotonicity is structural; it does not imply that a noisy operator choice was efficient.

\subsection{Finite-shot estimators and selection policies}
\label{sec:shots}

For the TFIM pools, $G_P$ has a real Pauli decomposition
\begin{equation}
G_P=\sum_w c_w W_w,\qquad c_w\in\mathbb{R}.
\end{equation}
Independent measurement of word $W_w$ with $N_w$ shots gives $\widehat{m}_w=2k_w/N_w-1$ and
\begin{equation}
\widehat{g}_P=\sum_wc_w\widehat{m}_w.
\label{eq:ghat}
\end{equation}
The simulator samples the exact binomial distribution.  Selection decisions use an empirical variance with a one-half pseudocount,
\begin{align}
\widetilde{p}_w&=\frac{k_w+1/2}{N_w+1},
&\widetilde{m}_w&=2\widetilde{p}_w-1,\\
\widehat{\sigma}_P^2
&=\sum_wc_w^2\frac{1-\widetilde{m}_w^2}{N_w+1}.
\label{eq:empiricalvariance}
\end{align}
This avoids zero estimated uncertainty when a finite sample happens to contain identical outcomes.

We compare three policies at a common base budget and ceiling per measured Pauli word:
\begin{enumerate}
\item \emph{Fixed}: measure every candidate once at the base budget.
\item \emph{Uniform escalation}: double cumulative shots for every remaining candidate until the statistical gates pass or the ceiling is reached.
\item \emph{Racing}: use the same cumulative doubling, but remove a candidate $i$ when its upper confidence bound is below the best lower bound.  This is a confidence-bound successive-elimination policy in the best-arm formulation of generator selection \cite{huang2025bestarm}.  With $\kappa=3$,
\begin{equation}
U_i=|\widehat g_i|+\kappa\widehat\sigma_i,
\qquad
L_i=\max(0,|\widehat g_i|-\kappa\widehat\sigma_i),
\end{equation}
and retain $i$ only if $U_i\geq\max_jL_j$.
\end{enumerate}
The leading candidate must first be statistically nonzero,
\begin{equation}
|\widehat g_{(1)}|\geq\kappa\widehat\sigma_{(1)},
\end{equation}
and, when rank resolution is required, satisfy
\begin{equation}
|\widehat g_{(1)}|-|\widehat g_{(2)}|
\geq\kappa\sqrt{\widehat\sigma_{(1)}^2+\widehat\sigma_{(2)}^2}.
\label{eq:rankgate}
\end{equation}
The deterministic ADAPT threshold is applied only after a statistically nonzero signal is established.  At the ceiling, an ambiguous but nonzero leader may be accepted.  These intervals are transparent heuristics, not simultaneous confidence guarantees: word covariances, multiple comparisons, and winner-selection bias are not corrected.

\section{Computational protocol}
\label{sec:protocol}

All calculations use $J=h=1$ and open boundaries.  Exact diagonalization provides validation targets but is not used to propagate variational states.  Sparse Pauli-word conjugation, Eq.~\eqref{eq:fastconjugation}, is used for forward states and backward observables.  The optimizer is L-BFGS-B with analytic gradients; a deterministic Adam fallback is included in the replication code.

For the HVA, we use $n=4,5,6$, depths $p=1,2,3$, and depth continuation: the optimized depth-$p$ parameters initialize the first $2p$ entries at depth $p+1$, with the new layer initialized to $(0.05,0.05)$.  We report relative energy error
\begin{equation}
\epsilon_E=\frac{|E-E_0|}{|E_0|},
\end{equation}
fidelity to the exact ground-state subspace, and peak Pauli supports encountered during the forward and backward sweeps.

Exact ADAPT uses gradient threshold $10^{-7}$.  The compact pool is allowed to run until convergence or pool exhaustion; the local-three pool is limited to 20 appended operators.  Finite-shot tests use $n=4$, the compact pool, at most eight operators, exact parameter reoptimization, and 100 independent seeds per configuration.  A run is counted as successful when $\epsilon_E<10^{-3}$.  Binomial success intervals are two-sided $95\%$ Wilson intervals \cite{wilson1927probable}.  The strategy comparison fixes the base at 128 shots and the ceiling at 4096 shots per Pauli word.  A separate racing sweep fixes the base at 64 and varies the ceiling over $512,2048,8192$, thereby separating the starting budget from the gradient-resolution floor.

The algebraic kernel is checked against dense matrices for products, adjoints, traces, Hilbert--Schmidt norms, Clifford anticommutation, inverse Jordan--Wigner decoding, fermionic anticommutation, gate unitarity, Bell-state correlations, Kraus completeness, and Trotter dynamics.  The algorithm tests additionally verify Eq.~\eqref{eq:adjointgradient}, Eq.~\eqref{eq:adaptgradient}, cumulative-sampling equivalence, estimator calibration, deterministic seeds, purity, $E\geq E_0$, and the fixed shot ceiling.

\section{Results}
\label{sec:results}

\subsection{Fixed-depth HVA}

Figure~\ref{fig:deterministic}(a) shows that increasing HVA depth improves the critical-point energy for all tested sizes, but a fixed depth does not maintain a size-independent error.  At $p=3$, the relative error rises from $4.84\times10^{-5}$ at $n=4$ to $3.67\times10^{-3}$ at $n=6$.  Table~\ref{tab:hva} also reports active supports.  The state support reaches half of the complete Pauli basis in these runs, while the backward observable remains substantially sparser.

\begin{figure*}[t]
\centering
\includegraphics[width=0.485\textwidth]{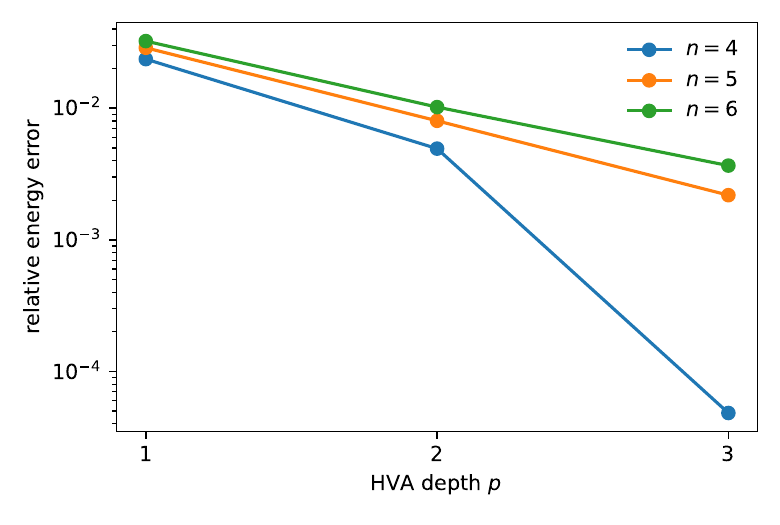}\hfill
\includegraphics[width=0.485\textwidth]{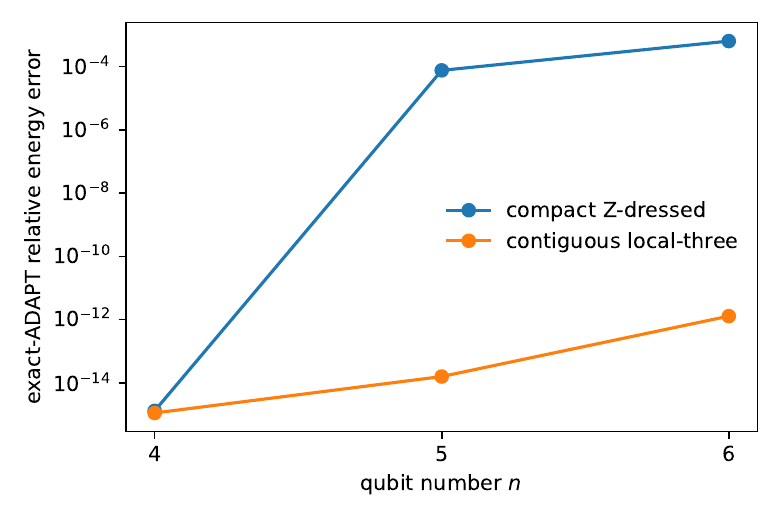}
\caption{Deterministic critical-point results.  (a) Relative energy error of the Hamiltonian variational ansatz versus depth.  (b) Exact-ADAPT error for the compact $\mathcal{A}_{\mathrm{ZYZ}}$ pool and the systematic $\mathcal{A}_{\mathrm{local3}}$ pool.  The local-three runs use a 20-operator cap.}
\label{fig:deterministic}
\end{figure*}

\begin{table}[t]
\caption{Depth-three HVA at $J=h=1$.  $S_\rho$ and $S_O$ are the largest Pauli supports in the forward-state and backward-observable sweeps.}
\label{tab:hva}
\centering
\begin{tabular}{ccccc}
\toprule
$n$ & $\epsilon_E$ & Fidelity & $S_\rho$ & $S_O$\\
\midrule
4 & $4.84\times10^{-5}$ & 0.999948 & 128 & 28\\
5 & $2.19\times10^{-3}$ & 0.995927 & 512 & 45\\
6 & $3.67\times10^{-3}$ & 0.988897 & 2048 & 66\\
\bottomrule
\end{tabular}
\end{table}

\subsection{Exact ADAPT and pool expressivity}

The compact pool is highly effective at $n=4$: six operators recover the numerical ground-state energy.  Its residual error becomes visible at $n=5$ and $n=6$, even though the selected state retains high fidelity.  This is a pool-expressivity limitation rather than a failure of the transpose-parity theorem; the theorem removes impossible candidates but does not guarantee that a small local subset spans a sufficiently rich tangent space.

The systematic local-three pool removes this stagnation.  It reaches $1.59\times10^{-14}$ at $n=5$ after 17 selections and $1.29\times10^{-12}$ at $n=6$ within the 20-operator budget.  The final $n=6$ state uses 1054 Pauli coefficients out of 4096.  Table~\ref{tab:exactadapt} summarizes the comparison.  At $n=4$, the seventh local-three angle is below $3\times10^{-8}$, so the numerically relevant circuit is effectively six-parameter, consistent with the compact result.

\begin{table}[t]
\caption{Exact ADAPT at the critical point.  ``Ops.'' is the number appended before the gradient threshold or 20-operator cap; $S_f$ is final state support.}
\label{tab:exactadapt}
\centering
\begin{tabular}{cccccc}
\toprule
$n$ & Pool & Size & Ops. & $\epsilon_E$ & $S_f$\\
\midrule
4 & compact & 12 & 6 & $1.31\times10^{-15}$ & 70\\
4 & local3 & 42 & 7 & $1.12\times10^{-15}$ & 70\\
5 & compact & 16 & 8 & $7.61\times10^{-5}$ & 240\\
5 & local3 & 60 & 17 & $1.59\times10^{-14}$ & 272\\
6 & compact & 20 & 10 & $6.41\times10^{-4}$ & 820\\
6 & local3 & 78 & 20 & $1.29\times10^{-12}$ & 1054\\
\bottomrule
\end{tabular}
\end{table}

\subsection{Finite-shot selection policies}

Figure~\ref{fig:strategies} and Table~\ref{tab:strategies} compare the three policies at a common base and ceiling.  A single fixed batch produces no successful runs: $0/100$, with Wilson interval $[0.0\%,3.7\%]$.  Uniform escalation and racing both produce $84/100$ successes, with interval $[75.6\%,89.9\%]$, and the same median relative error $3.11\times10^{-4}$.  Racing reduces the median measurement cost from $0.803$ million to $0.530$ million shots, a $34\%$ reduction, by ceasing to remeasure candidates whose confidence bounds no longer overlap the leaders.

\begin{figure*}[t]
\centering
\includegraphics[width=0.485\textwidth]{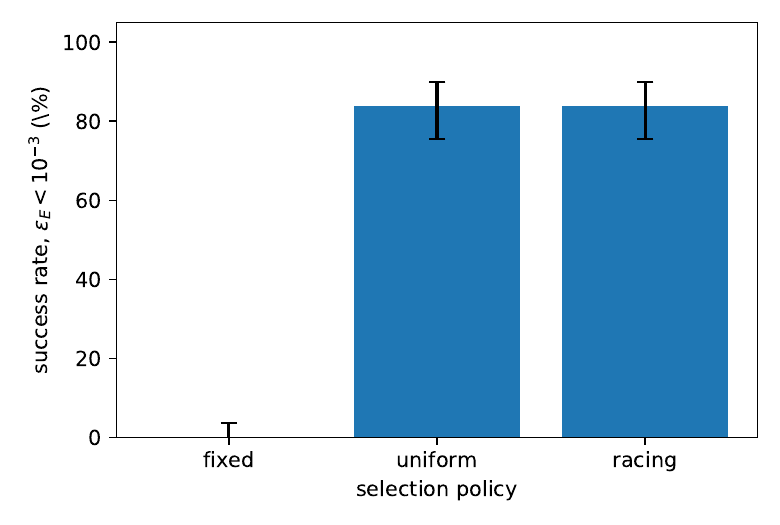}\hfill
\includegraphics[width=0.485\textwidth]{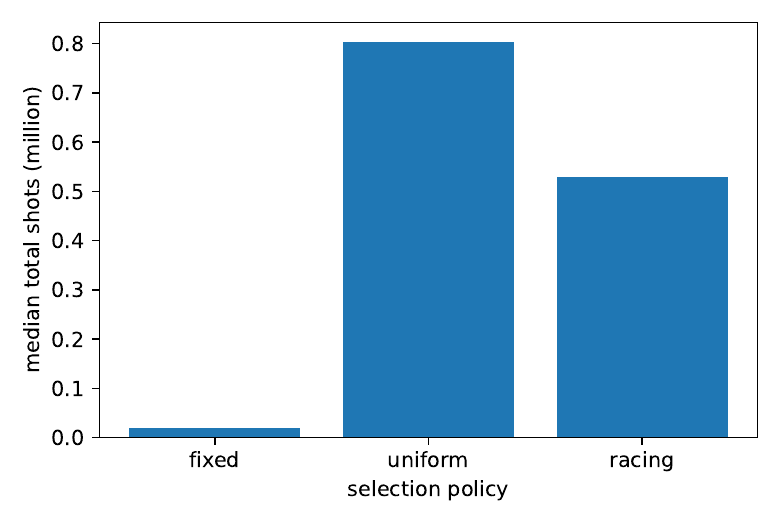}
\caption{Finite-shot policy comparison at $n=4$, base 128 and ceiling 4096 shots per Pauli word.  (a) Success probability for $\epsilon_E<10^{-3}$ with $95\%$ Wilson intervals.  (b) Median total measurements.  Uniform escalation and racing have equal observed success, while racing removes clear losers and lowers the median shot count.}
\label{fig:strategies}
\end{figure*}

\begin{table}[t]
\caption{Finite-shot policy comparison, 100 seeds.  CI is the $95\%$ Wilson interval.}
\label{tab:strategies}
\centering
\begin{tabular}{lccc}
\toprule
Policy & Success & CI (\%) & Median shots\\
\midrule
Fixed & $0/100$ & $[0.0,3.7]$ & 18,816\\
Uniform & $84/100$ & $[75.6,89.9]$ & 802,816\\
Racing & $84/100$ & $[75.6,89.9]$ & 530,176\\
\bottomrule
\end{tabular}
\end{table}

The ceiling sweep in Fig.~\ref{fig:ceiling} separates the initial budget from the finest resolvable gradient.  With the base held at 64, increasing the ceiling from 512 to 8192 shots raises success from $4\%$ to $89\%$ and lowers the median final error from $5.16\times10^{-3}$ to $5.15\times10^{-5}$, at the cost of approximately twenty times more measurements.  The statistical intervals remain wide enough to matter; Table~\ref{tab:ceiling} reports them explicitly.

\begin{figure*}[t]
\centering
\includegraphics[width=0.485\textwidth]{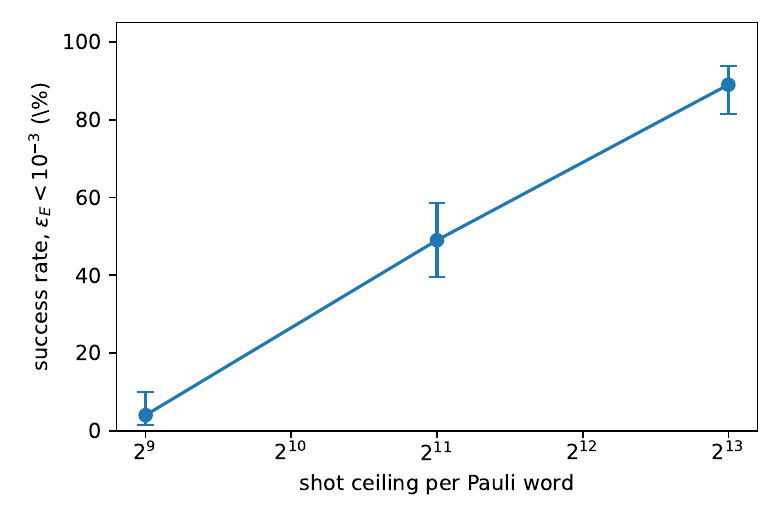}\hfill
\includegraphics[width=0.485\textwidth]{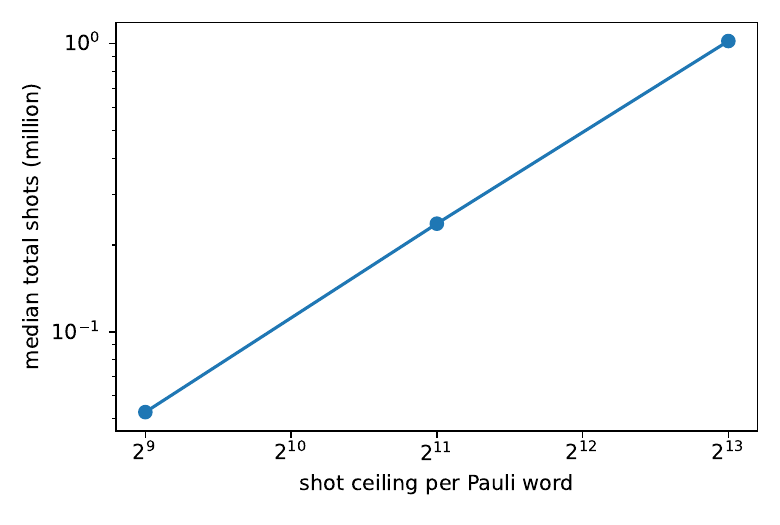}
\caption{Racing with a fixed base of 64 shots and a varying per-word ceiling.  (a) Success probability with $95\%$ Wilson intervals.  (b) Median total measurements.  The ceiling controls the smallest gradient separation that the ranking rule can resolve.}
\label{fig:ceiling}
\end{figure*}

\begin{table}[t]
\caption{Racing ceiling sweep, 100 seeds per configuration.}
\label{tab:ceiling}
\centering
\begin{tabular}{rccc}
\toprule
Ceiling & Success & CI (\%) & Median shots\\
\midrule
512 & $4/100$ & $[1.6,9.8]$ & 52,576\\
2048 & $49/100$ & $[39.4,58.7]$ & 236,992\\
8192 & $89/100$ & $[81.4,93.7]$ & 1,019,648\\
\bottomrule
\end{tabular}
\end{table}

Every tested run remains pure to numerical tolerance and obeys $E\geq E_0$.  These properties are expected because the finite-shot ablation perturbs only operator selection; state preparation and parameter reoptimization remain exact.  The frequent rank ambiguity is also expected at a reflection-symmetric point, where symmetry-related candidates can have equal exact gradients.  More shots cannot resolve an exact tie, so a practical algorithm must accept a symmetry-equivalent choice or compare equivalence classes.

\section{Discussion}
\label{sec:discussion}

The operator-centric language is most useful where it exposes a structural restriction or supports a sparse computation.  Proposition~\ref{prop:selectionrule} is one such restriction: the transpose action is diagonal on Pauli words, so an otherwise matrix-level real-symmetry argument becomes a one-line pool filter.  The same representation makes $G_P$ in Eq.~\eqref{eq:selectionobservable} an ordinary Hermitian multivector, so its exact value, finite-shot estimator, and uncertainty are handled without changing data structures.

The results also delineate what the formulation does \emph{not} provide.  The compact pool's stagnation at $n=5,6$ shows that symmetry compatibility is not expressivity.  Likewise, sparse Pauli support is workload dependent.  Clifford conjugations preserve word cardinality \cite{aaronson2004stabilizer}, and one Pauli rotation at most doubles it, but generic non-Clifford evolution eventually approaches dense $4^n$ support.  The $n=6$ HVA forward states already occupy 2048 words, while backward local observables remain much smaller.  A practical simulator should therefore switch representations according to support, not insist on sparse Clifford algebra in all regimes.

Our finite-shot implementation of confidence-bound racing is complementary to measurement grouping and reuse.  Anastasiou \emph{et al.} reduce the number of distinct commuting gradient measurements \cite{anastasiou2023gradients}; Ikhtiarudin \emph{et al.} reuse Pauli outcomes and allocate shots by variance \cite{ikhtiarudin2025shot}; and Huang and Izmaylov formulate generator selection as best-arm identification with successive elimination \cite{huang2025bestarm}.  The present ablation isolates the last mechanism at a fixed ceiling.  Combining all three is a natural next step: group commuting words, reuse compatible outcomes across candidates and iterations, then eliminate candidates using covariance-aware confidence regions.

The shared-parameter result has a broader methodological implication.  A layer can consist of mutually commuting Pauli gates while its aggregate generator has many eigenvalues.  The standard two-shift formula then does not apply to a simultaneous shift of the entire layer.  Applying parameter shift at the physical-gate level, or using the exact adjoint sweep, avoids a silent gradient error.

\section{Limitations}
\label{sec:limitations}

The numerical claims are deliberately bounded.  First, $n\leq6$ verifies correctness and qualitative behavior but does not establish favorable scaling.  Second, only ADAPT selection is sampled; optimizer shot noise, device decoherence, state-preparation and measurement error, calibration drift, and compilation constraints are absent.  Third, each Pauli word is measured independently.  Covariance between commuting words, simultaneous measurement, classical shadows \cite{huang2020shadows}, and reuse across candidates can lower the actual cost.  Fourth, the $3\sigma$ signal, gap, and racing rules are heuristics without familywise coverage, and the selected leader is subject to winner's bias.  Fifth, the no-repeat local pools are problem-specific greedy choices; exact reoptimization ensures monotonic energy but cannot refund a poor ansatz slot or the shots spent selecting it.  Finally, the dense exact eigensolver remains the validation backend, so the present experiments are not demonstrations of quantum advantage.

\section{Conclusion}

We have formulated fixed-depth and adaptive VQE in a sparse operator-centric realization of $\Cl(2n,\C)$.  The corrected Jordan--Wigner dictionary separates anticommuting Clifford generators from local Pauli operators; phase-normalized blades make the Pauli correspondence precise; and Pauli-word rotations are distinguished from Spin rotors.  The analytic organizing principle is a density-operator transpose-parity formulation of the established odd-$Y$ restriction for real-state ADAPT selection.  Numerically, the rule supports compact pools but does not replace an expressivity analysis: a systematic local-three pool is needed to maintain numerical accuracy through $n=6$.  Under finite-shot selection, cumulative escalation is essential, and confidence-bound racing retains the observed success of uniform escalation with one-third fewer median shots.  These results provide a reproducible baseline for future covariance-aware, measurement-reusing, and hardware-noisy adaptive eigensolvers.

\appendix

\section{Hermitian phase normalization of Clifford blades}
\label{app:blades}

For $A=\{a_1<\cdots<a_k\}$, let $q=k(k-1)/2$.  Reversing the ordered product requires $q$ transpositions, hence
\begin{equation}
(\gamma_{a_1}\cdots\gamma_{a_k})^\dagger
=(-1)^q\gamma_{a_1}\cdots\gamma_{a_k}.
\end{equation}
Using Eq.~\eqref{eq:hermitianblade},
\begin{align}
\widehat\gamma_A^\dagger
&=(-\ii)^q(-1)^q\gamma_{a_1}\cdots\gamma_{a_k}
=\ii^q\gamma_{a_1}\cdots\gamma_{a_k}
=\widehat\gamma_A,\\
\widehat\gamma_A^2
&=\ii^{2q}(-1)^q\mathbf{1}=\mathbf{1}.
\end{align}
In the Jordan--Wigner representation, $\widehat\gamma_A$ is therefore a Hermitian tensor product of one-qubit Paulis with square one, and hence equals a phase-free Pauli word up to sign.  This establishes the basis-level bijection used by the sparse kernel.

\section{Physical-gate parameter shifts}
\label{app:shift}

For one gate $G_k(\theta)=e^{-\ii\theta P_k/2}$, isolate the circuit as $U=U_{>k}G_kU_{<k}$.  Defining $\rho_{k-1}=U_{<k}\rho_0U_{<k}^\dagger$ and $O_k=U_{>k}^\dagger H U_{>k}$ gives
\begin{equation}
E=\Tr(O_kG_k\rho_{k-1}G_k^\dagger).
\end{equation}
Differentiation with Eq.~\eqref{eq:rotationderivative} yields the gate contribution in Eq.~\eqref{eq:adjointgradient}.  Since $P_k^2=1$, the same contribution equals
\begin{equation}
\frac12\left[E_k(\theta+\pi/2)-E_k(\theta-\pi/2)\right],
\end{equation}
where only physical gate $k$ is shifted.  If several gates share parameter $\theta_s$, the chain rule sums these gatewise expressions.  Simultaneously shifting all shared gates is a different two-point experiment and is exact only when the aggregate response has the required two-frequency structure.

\section{Reproducibility protocol}
\label{app:reproducibility}

The replication package contains the algebra kernel, HVA solver, exact and finite-shot ADAPT implementation, adversarial regression tests, raw per-seed records, summary tables, and figure-generation script.  Long experiments are separated into restartable stages:
\begin{verbatim}
python ga_qc_operator_improved.py
python vqe_tfim_rotor_better.py
python adapt_vqe_noisy_better_v2.py --seeds 3
python review_noisy_better.py
python pra_experiments.py --stage hva
python pra_experiments.py --stage adapt
python pra_experiments.py --stage shots \
  --workers 8
python pra_experiments.py --stage figures
\end{verbatim}
For deterministic numerical libraries, the supplied run script sets the BLAS thread count to one per worker.  Strategy seeds are generated from fixed integer schedules stored in the raw JSON.  All figure values in the manuscript are read from the accompanying CSV/JSON outputs rather than transcribed from console logs.  Exact sequences, optimized angles, and distribution quantiles are tabulated in the Supplemental Material.

\section*{Code and data availability}

The complete replication package, including source code, tests, raw and processed numerical data, and figure-generation scripts, is publicly available at \href{https://github.com/gutama/hkdipojono}{github.com/gutama/hkdipojono}.

\bibliographystyle{apsrev4-2}
\bibliography{references}

\end{document}


\begin{center}
{\Large Supplemental Material for ``Operator-centric Clifford algebra for variational eigensolvers and finite-shot adaptive selection''}\\[0.5em]
Ginanjar Utama and Hermawan Kresno Dipojono
\end{center}

\section*{S1. Verification matrix}
The released test suite checks the following independent invariants before generating manuscript results.
\begin{longtable}{p{0.25\textwidth}p{0.67\textwidth}}
\toprule
Layer & Checks \\
\midrule
Algebra kernel & Pauli multiplication, associativity fuzz tests, $(AB)^\dagger=B^\dagger A^\dagger$, trace and Hilbert--Schmidt inner product against dense matrices.\\
Clifford bridge & Jordan--Wigner anticommutation across qubits, inverse blade decoding, phase-normalized blade Hermiticity, and distinction between bivector Spin generators and higher-grade Pauli rotations.\\
Fermionic layer & Canonical anticommutation relations, nilpotency, and $c_j^\dagger c_j=(1-Z_j)/2$.\\
States and channels & Density-matrix positivity, purity, Bell/CHSH identities, partial trace, partial transpose, Kraus completeness, and channel fixed points.\\
HVA & Analytic adjoint gradients against finite differences and the sum of physical-gate parameter shifts; failure of the naive shared-parameter two-shift rule is detected explicitly.\\
ADAPT & Hermiticity of $G_P=-\ii[H,P]/2$, word-sum equality with the exact gradient, real-sector selection rule, exact final-state purity, and $E\geq E_0$.\\
Statistics & Incremental sampling equals one combined binomial batch, estimator mean and variance calibration, fixed shot ceiling, deterministic seed replay, and racing shot reduction.\\
\bottomrule
\end{longtable}

\section*{S2. Depth-three HVA parameters}
The main-text HVA results use depth continuation from $p=1$ to $p=3$.  Parameters are ordered $(\gamma_1,\beta_1,\gamma_2,\beta_2,\gamma_3,\beta_3)$.
\begin{table}[h]
\centering
\small
\begin{tabular}{c c c c c c c}
\toprule
$n$ & $\gamma_1$ & $\beta_1$ & $\gamma_2$ & $\beta_2$ & $\gamma_3$ & $\beta_3$\\
\midrule

4 & 0.411965876 & 1.413639501 & 0.893274926 & 1.211760454 & 0.848661266 & 0.513080679 \\
5 & 0.394638410 & 1.403131027 & 0.857848069 & 1.200648231 & 0.825190366 & 0.516382488 \\
6 & 0.389213004 & 1.399760380 & 0.844723474 & 1.198242677 & 0.816721402 & 0.519880428 \\
\bottomrule
\end{tabular}
\end{table}

\section*{S3. Exact ADAPT sequences}
The following tables report the exact selections used in the main text.  Angles are in radians.  Very small terminal angles are retained because the stopping rule acts on the next pool gradient, not by post hoc parameter pruning.

\subsection*{$n=4$, compact $\mathcal{A}_{\mathrm{ZYZ}}$}
\begin{longtable}{r l r}
\toprule
Step & Generator & Optimized angle\\
\midrule
\endhead

1 & \texttt{Y0Z1} & -7.8752347878e-01 \\
2 & \texttt{Z1Y2} & -5.1526042684e-01 \\
3 & \texttt{Z2Y3} & -7.1662466554e-01 \\
4 & \texttt{Y2Z3} & 2.2383263726e-01 \\
5 & \texttt{Z0Y1} & 3.1545423372e-01 \\
6 & \texttt{Y1Z2} & -9.2134787671e-02 \\
\bottomrule
\end{longtable}

\subsection*{$n=5$, local-three $\mathcal{A}_{\mathrm{local3}}$}
\begin{longtable}{r l r}
\toprule
Step & Generator & Optimized angle\\
\midrule
\endhead

1 & \texttt{YZ[0:2]} & -7.5795757966e-01 \\
2 & \texttt{ZY[1:3]} & -7.9560772401e-01 \\
3 & \texttt{ZY[2:4]} & -9.5207833408e-01 \\
4 & \texttt{ZY[3:5]} & -4.2745537986e-01 \\
5 & \texttt{YXZ[0:3]} & 2.1498615419e-02 \\
6 & \texttt{YZX[2:5]} & 5.3604800328e-01 \\
7 & \texttt{XZY[2:5]} & -1.9799328282e-01 \\
8 & \texttt{ZXY[1:4]} & 2.6675725886e-01 \\
9 & \texttt{ZXY[0:3]} & 3.0066183949e-02 \\
10 & \texttt{ZXY[2:5]} & -9.6198797715e-03 \\
11 & \texttt{ZY[0:2]} & 2.6523642835e-01 \\
12 & \texttt{YZX[0:3]} & 4.5480416144e-08 \\
13 & \texttt{YXZ[2:5]} & -2.3526037628e-07 \\
14 & \texttt{XZY[1:4]} & -5.3773517271e-08 \\
15 & \texttt{YXZ[1:4]} & 1.0606885331e-07 \\
16 & \texttt{XYZ[0:3]} & 3.7200490912e-08 \\
17 & \texttt{XYZ[2:5]} & -2.2508359193e-08 \\
\bottomrule
\end{longtable}

\subsection*{$n=6$, local-three $\mathcal{A}_{\mathrm{local3}}$}
\begin{longtable}{r l r}
\toprule
Step & Generator & Optimized angle\\
\midrule
\endhead

1 & \texttt{YZ[0:2]} & -5.9381781767e-01 \\
2 & \texttt{ZY[1:3]} & -4.5158624676e-01 \\
3 & \texttt{ZY[2:4]} & -5.7595630997e-01 \\
4 & \texttt{ZY[3:5]} & -5.3987734575e-01 \\
5 & \texttt{ZY[4:6]} & -6.0420688685e-01 \\
6 & \texttt{YXZ[0:3]} & -3.0045396060e-01 \\
7 & \texttt{ZXY[3:6]} & -3.5783228215e-01 \\
8 & \texttt{XZY[1:4]} & -2.2991511571e-06 \\
9 & \texttt{YZX[2:5]} & -1.4102070170e-08 \\
10 & \texttt{YZ[3:5]} & 1.1790384432e-01 \\
11 & \texttt{YXZ[1:4]} & -1.2120425611e-01 \\
12 & \texttt{ZXY[0:3]} & 1.8448864773e-01 \\
13 & \texttt{ZXY[1:4]} & 3.1688862932e-03 \\
14 & \texttt{YXZ[3:6]} & 2.6001232788e-01 \\
15 & \texttt{YZ[2:4]} & 9.4207441269e-02 \\
16 & \texttt{YXZ[2:5]} & 1.0239248644e-01 \\
17 & \texttt{XYZ[1:4]} & 5.5738049604e-06 \\
18 & \texttt{YZX[0:3]} & -1.6313960480e-06 \\
19 & \texttt{ZXY[2:5]} & -2.0143451972e-01 \\
20 & \texttt{XZY[0:3]} & -3.6017973541e-07 \\
\bottomrule
\end{longtable}

\section*{S4. Finite-shot distribution summaries}
All configurations use 100 independent seeds.  The interquartile range (IQR) is reported in addition to the median shown in the main text.
\begin{longtable}{l r r r r r}
\toprule
Experiment & Base & Ceiling & Median $\epsilon_E$ & IQR $\epsilon_E$ & IQR shots\\
\midrule
\endhead

racing ceiling 512 & 64 & 512 & 5.157e-03 & [2.822e-03,8.849e-03] & [50,480,64,352] \\
racing ceiling 2048 & 64 & 2048 & 2.822e-03 & [5.146e-05,2.822e-03] & [220,752,271,952] \\
racing ceiling 8192 & 64 & 8192 & 5.146e-05 & [5.146e-05,3.111e-04] & [987,216,1,052,224] \\
fixed & 128 & 4096 & 2.221e-02 & [5.157e-03,2.221e-02] & [18,560,18,816] \\
racing & 128 & 4096 & 3.111e-04 & [5.146e-05,3.111e-04] & [517,856,534,528] \\
uniform & 128 & 4096 & 3.111e-04 & [5.146e-05,3.111e-04] & [743,424,815,104] \\
\bottomrule
\end{longtable}
The observed mean fractions of selections within $95\%$ of the exact best gradient and the mean fractions accepted with unresolved top-two rank are:
\begin{longtable}{l r r}
\toprule
Configuration & Near-optimal fraction & Ambiguous fraction\\
\midrule

racing/$N_{\max}=512$ & 0.948 & 0.997 \\
racing/$N_{\max}=2048$ & 0.985 & 0.957 \\
racing/$N_{\max}=8192$ & 0.975 & 0.794 \\
fixed & 0.902 & 0.997 \\
racing & 0.974 & 0.904 \\
uniform & 0.979 & 0.880 \\
\bottomrule
\end{longtable}

\section*{S5. Seed and measurement protocol}
For the strategy comparison, seed index $s=0,\ldots,99$ is mapped to
\[
10^6+10^4 q+s,
\]
where $q=0,1,2$ denotes fixed, uniform, and racing.  For the ceiling sweep,
\[
2\times10^6+10^4 q+s,
\]
where $q=0,1,2$ denotes ceilings $512,2048,8192$.  Each measured Pauli word is sampled independently from its exact binomial distribution.  Cumulative escalation adds only the samples needed to reach the next target budget; previously observed outcomes are never discarded.  Raw counts, selected sequences, energies, fidelities, purities, shot totals, and stopping reasons are included in \texttt{paper\_results/shot\_raw.json}.